\documentclass[a4paper]{article}

\usepackage{authblk}
\usepackage{fullpage}
\usepackage{mathtools}
\usepackage{amssymb}
\usepackage{amsthm}
\usepackage{tikz}
\usepackage{paralist}
\usepackage{enumitem}
\usepackage{xspace}
\usepackage[basic,classfont=caps]{complexity}
\usepackage{todonotes}
\renewclass{\DTIME}{Dtime}
\renewclass{\DSPACE}{Dspace}
\renewclass{\EXP}{Exptime}
\newclass{\TWOEXP}{2Exptime}
\newclass{\ACK}{Ack}
\newclass{\FDTIME}{FDtime}

\makeatletter
\let\epsilon\varepsilon
\let\phi\varphi
\makeatother

\usetikzlibrary{positioning,fit,arrows,automata,calc,shapes} 
\tikzset{->,>=stealth',shorten >=1pt,shorten <=1pt,auto,node distance=1.5cm,
every loop/.style={looseness=6},
initial text={},
el/.style={font=\scriptsize},
inner sep=0.5mm,
}

\newtheorem{lemma}{Lemma}
\newtheorem{proposition}{Proposition}
\newtheorem{theorem}{Theorem}
\newtheorem{problem}{Problem}
\newtheorem{definition}{Definition}

\newcommand{\defeq}{:=}
\newcommand{\st}{\mathrel{\mid}}
\newcommand{\pow}[1]{\mathcal{P}(#1)}
\newcommand{\eve}{Eve\xspace}
\newcommand{\adam}{Adam\xspace}
\newcommand{\calM}{\mathcal{M}}
\newcommand{\Obs}{\mathsf{Obs}}
\newcommand{\post}{\mathsf{post}}
\newcommand{\supp}{\mathsf{supp}}
\newcommand{\obs}{\mathsf{obs}}
\newcommand{\el}{\mathsf{EL}}
\newcommand{\posEnergy}{\mathsf{PosEn}}
\newcommand{\safe}{\mathsf{Safe}}
\newcommand{\Ack}{F_\omega}
\newcommand{\plays}{\mathsf{Plays}}
\newcommand{\prefs}{\mathsf{Prefs}}

\begin{document}

\title{The Fixed Initial Credit Problem for Partial-Observation Energy Games is
\ACK-complete}
\author{Guillermo A. P\'erez\thanks{Author supported by an F.R.S.-FNRS fellowship.}}
\affil{D\'{e}partament d'Informatique, Universit\'{e} Libre de Bruxelles\\
\texttt{gperezme@ulb.ac.be}}

\maketitle

\begin{abstract}
	In this paper we study two-player games with asymmetric
	partial observation and an energy objective. Such games are played on a
	weighted automaton by \eve, choosing actions, and \adam, choosing a
	transition labelled with the given action. \eve attempts to maintain the
	sum of the weights (of the transitions taken) non-negative while \adam
	tries to do the opposite. \eve does not know the exact state of the
	game, she is only given an equivalence class of states which contains
	it. In contrast, \adam has full observation.  We show the fixed initial
	credit problem for these games is \ACK-complete.
\end{abstract}

\section{Introduction}
Energy games are two-player quantitative games of infinite horizon played on
finite weighted automata. The game is played in rounds in which one player, \eve,
chooses letters from the automaton's alphabet whilst \adam, the second player,
resolves non-determinism. The initial configuration of the game is determined by
an automaton, an initial state, and an \emph{initial credit} for \eve. The goal
of \eve in an energy game is to keep a certain resource from being depleted. More
specifically, she wins if, for every round, the sum of the weights of the
transitions traversed so far plus her initial credit is non-negative. \adam has
the opposite objective: witnessing a negative value.

Quantitative games, in general, are useful models for the interaction of the
controller and its environment in open reactive systems. Furthermore,
synthesizing controllers in this setting reduces to computing a winning strategy
for one of the players in the corresponding game. Energy games in particular are
useful for systems in which one is interested in the use of bounded resources
such as power or fuel~\cite{cd-ahs03,bflms08}.

Two decision problems for energy games have been studied by the formal
verification community: the \emph{fixed initial credit} and \emph{unknown
initial credit problems}. The former asks whether, given a fixed initial credit
for \eve, she has a strategy which ensures all plays consistent with it are
winning. The latter is more ambitious in that it asks whether there exists some
initial credit for which the same question has a positive answer. It is known
that if \eve has a winning strategy in an energy game, she also has a memoryless
winning strategy. Furthermore, in order to win, \eve essentially has to ensure
staying in cycles with non-negative (total) weight. Using these two facts, one
can show the unknown initial credit reduces in polynomial time to the fixed
initial problem. (If there is some initial credit for which \eve wins, then 
$n w_{\max}$ should suffice---where $n$ is the number of states and $w_{\max}$
denotes the maximum absolute value of a transition weight in the automaton.) It
is also known that the fixed initial credit problem is log-space equivalent to
the threshold problem for \emph{mean-payoff games}~\cite{bflms08}. In a
mean-payoff game, the objective of \eve consists in maximizing the limit
(inferior) of the averages of the running sum of transition weights observed
along an infinite play. Determining if \eve can ensure at least a given
mean-payoff value is not known to be decidable in polynomial time. However, as
they are known to be positionally determined~\cite{zp96}, the latter problem can
be decided in \NP\ and in \coNP. This places energy games in a very special
class of problems which are known to be solvable in $\NP \cap \coNP$ and for
which no polynomial-time algorithm has been
discovered.\footnote{Other games known to be in the same class are
	discounted-sum and simple stochastic games~\cite{zp96}. In fact, all
	games mentioned in this work have further been shown to be in $\UP \cap
	\coUP$~\cite{jurdzinski98} and are therefore unlikely to be
	\NP-complete.}

Multi-dimensional energy games have been studied in, amongst other
works,~\cite{cdhr10} and more recently in~\cite{jls15}. These are the
natural generalization of energy games played on singly-weighted automata to
automata with vector weights on their transitions. They are relevant to the
synthesis of reactive controllers sensitive to the usage of multiple resources.
In this setting, it has been shown that the fixed initial credit problem is
\TWOEXP-complete while the unknown initial credit problem is simpler, namely,
\coNP-complete.

In the present work we focus on another generalization, namely, \emph{energy
games with partial observation}. Such games are, once more, played on finite
singly-weighted automata. The difference between classical energy games
and partial-observation energy games is that, in the latter, a partition of the
states of the automaton into \emph{observations} is also given as part of the
input. The game is then modified so that, after every round, \eve is only
informed of the observation of the successor state chosen
by \adam (hence the partial observability). Energy games with
partial observation where initially studied in~\cite{ddgrt10} and~\cite{HPR14}.
They capture the fact that controllers in open reactive systems have limited
capabilities, e.g. a finite number of sensors with limited precision. Two
results from~\cite{ddgrt10} are of particular interest to us.  First, the
unknown initial credit problem for partial-observation energy games was shown to
be undecidable by reduction from the halting problem for Minsky machines.
Second, decidability of the fixed initial credit problem was established by
describing a reduction to finite safety games.

\paragraph*{Contributions}
In this work we refine the upper and lower bounds for the fixed initial credit
problem for partial-observation energy games. Specifically, we show the
size of the safety game used in the algorithm from~\cite{ddgrt10} is at most
Ackermannian with respect to the size of the input game. We then describe how
the Minsky machine simulation used to show undecidability of the unknown initial
credit problem in~\cite{ddgrt10} can be modified to show \ACK-hardness of the
fixed initial credit problem. This establishes \ACK-completeness of the problem.

%\paragraph*{Further related works}
%It is known that a lower bound for \emph{blind}\footnote{A partial-observation
%game is blind if all states have the same observation.} quantitative games with
%objective $\Omega$ translates into a lower bound for the \emph{universality}
%problem for quantitative automata with the dual acceptance condition (i.e. an
%infinite word $x$ is accepted by the automaton if it has a run on $x$ with the
%property $\lnot \Omega$). Conversely, an upper bound for the universality
%problem in a specific type of quantitative automata also implies the same bound
%for the corresponding type of games. For instance, in~\cite{cdehr10} the
%undecidability result for blind mean-payoff games from~\cite{ddgrt10} was used
%to claim the universality problem for mean-payoff automata is also undecidable.
%However, it is important to note that, unlike mean-payoff, the energy objective
%is not the dual of itself. It follows that the \ACK-completeness result
%presently obtained indeed implies \ACK-hardness of the universality problem
%for a specific type of quantitative automata but not for ``energy automata''. It
%is interesting to note that energy automata have indeed been studied in the form
%of \emph{one-counter nets} in~\cite{ht14} and that universality for this kind of
%automata is actually \ACK-complete as well.

\section{Preliminaries}\label{sec:preliminaries}
\paragraph*{Games}
A \emph{weighted game with partial observation} (or just game, for short) is a
tuple $( Q,q_0,\Sigma,\Delta,w,\Obs )$ where $Q$ is a non-empty finite set of
states, $q_0 \in Q$ is the initial state, $\Sigma$ is a finite alphabet of
actions or symbols, $\Delta \subseteq Q \times \Sigma \times Q$ is a total
transition relation, $w : \Delta \to \mathbb{Z}$ is a weight function, and $\Obs
\subseteq \pow{Q}$ is a partition of $Q$ into \emph{observations}. If $\Obs =
\{Q\}$ we say the game is \emph{blind}, if $\Obs = \{\{q\} \st q \in Q\}$ we say
it is a \emph{full-observation} game. For $s \subseteq Q$ and $\sigma \in
\Sigma$, denote by
$\post_\sigma(s) \defeq \{ q' \in Q \st \exists q \in s \land (q,\sigma,q') \in
\Delta\}$ the set of $\sigma$-successors of $s$. Also, let $w_{\max}$ denote the
maximum absolute value of a transition weight in the automaton.

Unless otherwise stated, in what follows we consider a fixed $G = ( Q,
q_0,\Sigma,\Delta,w,\Obs)$.

\paragraph*{Plays}
A \emph{play} in $G$ is an infinite sequence $\pi = q_0 \sigma_0 q_1 \dots$ such
that $(q_i,\sigma_i,q_{i+1}) \in \Delta$, for all $i \ge 0$. For a play $\pi$ and
integers $0 \le i \le j$, we denote by $\pi[i..j]$ the infix $q_i a_i \dots
a_{j-1} q_j$ of
$\pi$. The set of plays in $G$ is denoted by $\plays(G)$ and the set of
prefixes of plays ending in a state is written $\prefs(G)$. The unique
observation containing state $q$ is denoted by $\obs(q)$. We extend
$\obs(\cdot)$ to plays and prefixes in the natural way. For instance, we obtain
the \emph{observation sequence} $\obs(\pi)$ of a play $\pi$ as follows: $\obs(q_0)
\sigma_0 \obs(q_1) \sigma_1 \dots$

\paragraph*{Objectives}
An objective in a game corresponds to a set of ``good'' plays.

The \emph{energy level} of a play prefix $\pi = q_0 \sigma_0 \dots q_n$ is
$\el(\pi) = \sum^{n-1}_{i=0}w(q_i,\sigma_i,q_{i+1})$. The \emph{energy
objective} is parameterized by an initial credit $c_0 \in \mathbb{N}$ and is
defined as:
\[
	\posEnergy(c_0) \defeq \{ \pi \in \plays(G) \st \forall i > 0 : c_0 +
	\el(\pi[0..i]) \ge 0\}.
\]
In other words, the energy objective asks for the energy level of a play never to
drop below $0$ when starting with energy level $c_0$.

We will also make use of the \emph{safety} objective, defined relative to a set
of \emph{safe states} $W \subseteq Q$:
\[
	\safe(W) \defeq \{ q_0 \sigma_0 q_1 \dots \in \plays(G) \st \forall i \ge 0 :
	q_i \in W \}.
\]
Intuitively, the objective asks that \eve make sure the play never visits states
outside of $W$. Note that in safety games the weight function is not needed.

\paragraph*{Strategies}
A \emph{strategy for \eve} is a function $\lambda: \prefs(G) \to \Sigma$. A
strategy $\lambda$ for \eve is \emph{observation-based} if for all prefixes
$\pi,\pi' \in \prefs(G)$, if $\obs(\pi) = \obs(\pi')$ then $\lambda(\pi) =
\lambda(\pi')$. A prefix (or play) $\pi = q_0 \sigma_0 \dots$ is consistent with
a strategy $\lambda$ for \eve if $\lambda(\pi[0..i]) = \sigma_i$, for all $i \ge 0$.
We say a strategy $\lambda$ for \eve is \emph{a winning strategy for her} in a
game with objective $\Omega$ if all plays consistent with $\lambda$ are in
$\Omega$.

We do not formalize the notion of strategy for \adam here. Intuitively, given a
play prefix and an action $\sigma \in \Sigma$, he selects a $\sigma$-successor
$q$ of the current state and reveals $\obs(q)$ to \eve.

\begin{problem}[Fixed initial credit problem]\label{prob:fixed}
	Given a game $G$ and an initial credit $c_0$, decide whether
	there exists a winning observation-based strategy for \eve for the
	objective $\posEnergy(c_0)$.
\end{problem}

\section{Upper bound}\label{sec:upper}
The fixed initial credit problem was shown to be decidable in~\cite{ddgrt10}.
To do so, the problem is reduced to determining the winner of a
safety game played on a finite tree whose nodes are functions which
encode the \emph{belief} of \eve in the original game. The notion of belief
corresponds to the information \eve has about the current state (at any round) of
a partial-observation game. In this particular case, the belief of \eve is
defined by a prefix $\pi = q_0 \sigma_0 \dots \sigma_{n-1} q_n$ with $o =
\obs(q_n)$. It corresponds to the subset of states $s \subseteq o$ which are
reachable from $q_0$ via a prefix $\pi'$ with the same observation sequence as
$\pi$, i.e.  $\obs(\pi) = \obs(\pi')$, together with the energy levels of all
prefixes ending in $q$ for all $q \in s$. Note that \eve only really cares about
the minimal energy levels of prefixes ending in states from $s$. This
information can thus be encoded into functions.

In this section we will first formalize the construction described above. We
will then give an alternative argument (to the one presented in~\cite{ddgrt10})
which proves that the constructed tree is finite. The latter goes via a
translation from functions to vectors and prepares the reader for the next
result.  Finally, we will define an Ackermannian function and show that the size
of the tree is at most the value of the function on the size of the input game.

Throughout the section we consider a fixed partial-observation energy game $G =
(Q,q_0,\Sigma,\Delta,w,\Obs)$ and fixed initial credit $c_0 \in \mathbb{N}$.

\subsection{Reduction to safety game}

\paragraph*{Belief functions}
We define the set of \emph{belief
functions} of \eve as $\mathcal{F} \defeq \{f : Q \to \mathbb{Z} \cup \{\bot\}\}$.
The \emph{support} of a function $f \in \mathcal{F}$ is the set $\{q \in Q \st
f(q) \neq \bot\}$. A function $f \in \mathcal{F}$ is said to be negative if
$f(q) < 0$ for some $q \in \supp(f)$. The initial belief function $f_0$ has
support $\{q_0\}$ and $f_0(q_0) = c_0$. Given two functions $f,g \in
\mathcal{F}$ we define the order $f \preceq g$ to hold if $\supp(f) = \supp(g)$
and $f(q) \leq g(q)$ for all $q \in \supp(f)$. Additionally, for $\sigma \in
\Sigma$ we say $g$ is a \emph{$\sigma$-successor} of $f$ if $\exists o \in \Obs
: \supp(g) = \post_\sigma(\supp(f)) \cap o$ and $g(q) = \min\{f(p) +
w(p,\sigma,q) \st p \in \supp(f) \land (p,\sigma,q) \in \Delta\}$ for all $q \in
\supp(g)$. Intuitively, if \eve has belief function $f$ and she plays $\sigma$,
then if \adam reveals observation $o$ to her as the observation of the new state
of the game, she now has belief function $g$.

\paragraph*{Function-action sequences}
For a function-action sequence $s = f_0 \sigma_0 f_1 \dots \sigma_{n-1} f_n$ we
will write $f_s$ to denote $f_n$, i.e. the last function of the sequence. Let
$S$ be the smallest subset of $(\mathcal{F} \cdot \Sigma)^* \mathcal{F}$
containing $f_0$ and $s \cdot \sigma \cdot f$ if $s \in S$, $f$ is a
$\sigma$-successor of $f_s$, and it holds that:
\begin{inparaenum}[$(a)$]
	\item $f_s$ is not negative and
	\item $f_{s'} \not\preceq f_s$ for all proper prefixes $s'$ of $s$.
\end{inparaenum}
The desired full-observation safety game is then $H =
(S,f_0,\Sigma,E,W)$ where
\begin{itemize}[nolistsep]
	\item the transition relation $E \subseteq S \times \Sigma \times S$
		contains triples $(s,\sigma,s')$ where $s' = s \cdot
		\sigma \cdot f_{s'}$, and
	\item the safe states are $W = \{s \in S \st f_s \text{ is
		not negative}\}$.
\end{itemize}
In order for $E$ to be total, we add self-loops $(s,\sigma,s)$ for any $s \in S$
without outgoing transitions.

\begin{lemma}[From~\cite{ddgrt10}]\label{lem:safety-energy}
	There is a winning observation-based strategy for \eve for the objective
	$\posEnergy(c_0)$ in $G$ if and only if there is a winning strategy for
	\eve in the safety game $H$.
\end{lemma}

\subsection{Showing the safety game is finite}
Henceforth, let $H = (S,f_0,\Sigma,E,W)$ be the safety game constructed from the
partial-observation game $G$ and initial credit $c_0$ we have fixed for all of
Section~\ref{sec:upper}.

In the sequel, it will be useful to consider vectors instead of functions. We
will therefore define an encoding of belief functions into vectors. Formally,
let us fix two bijective mappings $\alpha : \{1,\dots,|Q|\} \to Q$ and $\beta :
\pow{Q} \to \{1,\dots,2^{|Q|}\}$. (The latter two mappings essentially
corresponding to fixing an ordering on $Q$ and the set of subsets of $Q$.) For a
belief function $f$, we will define a vector $\vec{f} \in \mathbb{Z}^{|Q|+2}$
which holds in its $i$-th dimension the value assigned by $f$ to state
$\alpha(i)$ in its support. For technical reasons (see
Lemma~\ref{lem:control}), if state $\alpha(i)$ is not part of the support
of $f$, we will use as place-holder the minimal value assigned by $f$ to any
state. Additionally, we use two dimensions to identify uniquely the support set
of $f$. More formally, the vector $\vec{f}$ is
\[
	\left(
	2^{|Q|} - \beta\left(\supp(f)\right),
	\beta\left(\supp(f)\right),
	\gamma \circ \alpha(|Q|),
	\dots,
	\gamma \circ \alpha(1)
	\right)
\]
where $\gamma(q)$ is $f(q)$ if $q \in \supp(f)$ and $\min\{f(q) \st q \in
\supp(f)\}$ otherwise.

It follows directly from the above definitions that two belief functions being
$\preceq$-comparable is sufficient and necessary for their corresponding vectors
to be $\le$-comparable.\footnote{To be precise, $\le$ here denotes the product
ordering on vectors of integers.} More formally,
\begin{lemma}\label{lem:fun-eq-vec}
	For belief functions $f,g \in \mathcal{F}$ we have that $f \preceq g$ if
	and only if $\vec{f} \le \vec{g}$.
\end{lemma}

Using the above Lemma we can already argue that the safety game $H$ is finite.
Suppose, towards a contradiction, that $H$ is infinite. By K\"onig's Lemma,
there is an infinite function-action sequence $s = f_0 \sigma_0 f_1 \dots$ such
that for all $i \ge 0$ we have:
\begin{inparaenum}[$(i)$]
	\item $f_{i+1}$ is a $\sigma_i$-successor of $f_i$,
	\item $f_i$ is not negative, and
	\item $f_j \not\preceq f_i$, for all $0 \le j < i$.
\end{inparaenum}
Now, let us consider the vector sequence $v = \vec{f_0} \vec{f_1}
\dots$. Note that the two dimensions used to represent the support of the
function cannot be negative. Hence, together with condition $(ii)$ above,
it follows that $\vec{f_i} \in \mathbb{N}^{|Q|+2}$ for all $i \ge 0$. That is,
the vectors have no negative integers. Further, from $(iii)$ together with
Lemma~\ref{lem:fun-eq-vec} it follows that there are no distinct vectors
$\vec{f_i}, \vec{f_j}$ in the sequence $v$ such that $\vec{f_i} \le \vec{f_j}$.
Thus, we get a contradiction with Dickson's Lemma---which states that
every infinite sequence of vectors of natural numbers has two distinct
$\le$-comparable elements.
\begin{proposition}[From~\cite{ddgrt10}]
	The game $H$ has finite state space.
\end{proposition}

\subsection{An Ackermann bound on the size of the safety game}
The argument we have presented, to show the safety game is finite, carries the
intuition that any function-action sequence from $H$ should eventually end in a
negative function or a function which is bigger than another
function in the sequence (w.r.t. $\preceq$).
How long can such sequences be? Using the relation
between functions and vectors that we have established (and formalized in
Lemma~\ref{lem:fun-eq-vec}) we can apply results of Schmitz et
al.~\cite{ffss11,ss12,schmitz16} which have been formulated for sequences of
vectors of natural numbers. Intuitively, the bound they provide is based on
``how big the jump is'' from each vector in the sequence to the next. This last
notion is formalized in the following definition.
\begin{definition}[Controlled vector sequence]
	A vector sequence $\mathbf{a_0} \mathbf{a_1} \dots \in
	\left(\mathbb{N}^d\right)^*$ is $t$-controlled by a unary increasing
	function $\kappa : \mathbb{N} \to \mathbb{N}$ if $|\mathbf{a_i}|_\infty
	< \kappa(t + i)$ for all $i \ge 0$.\footnote{For a vector $\mathbf{a} =
	(a_d,a_{d-1},\dots,a_1)$, the infinity norm is the maximum value on any
	dimension, i.e.  $\max\{a_i \st 1 \le i \le d\}$.}
\end{definition}

We will now define a hierarchy of sets of functions. Our intention is to
determine at which level of this hierarchy we can find a function which can be
said to control vector sequences induced by function-action sequences from $H$.
This will allow us to find a function---also at a specific level of the
hierarchy---that bounds the length of such sequences (see
Lemma~\ref{lem:length-function}).

\paragraph*{The fast-growing functions}
These can be seen as a
sequence $(F_i)_{i \ge 0}$ of number-theoretic functions defined inductively
below.~\cite{fw98} 
\begin{align*}
	F_0(x) &\defeq x + 1\\
	F_{i+1}(x) &\defeq F_i^{x +1}(x) = \overbrace{F_i(F_i(\dots
	F_i}^{x+1\text{ times}}(x)\dots))
\end{align*}

The following Lemma summarizes some
properties of the hierarchy.
\begin{lemma}[From~\cite{fw98}]\label{lem:props-fast-growing-functions}
	For all $i \in \mathbb{N}$,
	\begin{itemize}[nolistsep]
		\item for all $i \le j$ and $0 \le n \le m$, $F_j(m) \ge F_i(n)$
			and the latter is strict if the inequality between $j$
			and $i$ or $m$ and $n$ is strict;
		\item $F_i$ is primitive-recursive;
		\item $F_i$ is dominated by $F_{i+1}$.\footnote{For two
				functions $f,g : \mathbb{N} \to \mathbb{N}$, we
				say $g$ \emph{dominates} $f$ if $g(x) \ge f(x)$
				for all but finitely many $x \in \mathbb{N}$.}
	\end{itemize}
	Furthermore, for all primitive-recursive functions $f$, there exists $i
	\in \mathbb{N}$ such that $F_i$ dominates $f$.
\end{lemma}

We consider the following variant of the Ackermann function $\Ack(x) \defeq
F_x(x)$. It is not hard to show that $\Ack$ dominates all $F_i$---that is, for
all $i \in \mathbb{N}$---and, in turn, all primitive-recursive functions.

\paragraph*{The Grzegorczyk hierarchy}
We now introduce a sequence $(\mathfrak{F}_i)_{i \ge 2}$ of sets of
functions. Using the $i$-th fast-growing function, we define the $i$-th level of
the hierarchy~\cite{wainer70,schmitz16} as follows:
\[ \textstyle
	\mathfrak{F}_i \defeq \bigcup_{c \in \mathbb{N}} \FDTIME\left(
	F^c_i(x)
	\right).
\]
In other words, $\mathfrak{F}_i$ consists of all functions $\mathbb{N} \to
\mathbb{N}$ which can be computed by a deterministic Turing machine in time
bounded by any finite composition of the function $F_i$.  Note that, since
$F_2$ is of exponential growth, we could restrict space instead of time or
even allow non-determinism and obtain exactly the same classes.

The following property of the classes of functions from the hierarchy will be
useful.
\begin{lemma}[From~\cite{lw70,schmitz16}]
	For all $i \ge 2$, every $f \in \mathfrak{F}_i$ is dominated by $F_{j}$
	if $i < j$.
\end{lemma}

We will now show how to control vector sequences induced by function-action
sequences from $H$ (following the Karp-Miller tree
analysis from~\cite{ffss11}).
\begin{lemma}\label{lem:control}
	For all non-negative function-action sequences $s \in S$, the
	corresponding vector sequence from $\left(\mathbb{N}^{|Q|+2}\right)^*$
	is $(c_0 + w_{\max} + |Q|)$-controlled by $k(x) \defeq 2^x + x^2$.
\end{lemma}
\begin{proof}
	Let us assume that $w_{\max} > 0$. (This is no loss of generality as
	the energy game is trivial otherwise.) For any sequence $s = f_0
	\sigma_0 \dots \sigma_{n-1} f_n \in S$ we have that for all $0 \le i \le
	n$:
	\begin{align*}
		|\vec{f_{i}}|_\infty &\le \max\{2^{|Q|}, c_0 + i \cdot
		w_{\max}\}\\
		&\le 2^{|Q|} + (c_0 + w_{\max} + i)^2\\
		&\le 2^{|Q| + c_0 + w_{\max} + i} + (|Q| + c_0 + w_{\max} +
		i)^2
	\end{align*}
	which concludes the proof.
\end{proof}

Note that the control function from Lemma~\ref{lem:control} is at the second
level of the Grzegorczyk hierarchy. That is, $k \in \mathfrak{F}_2$, since
$F_2$ is exponential. We can now apply the following tool.
\begin{lemma}[From~\cite{ffss11}]\label{lem:length-function}
	For natural numbers $d,i \ge 1$, for all unary increasing functions
	$\kappa \in \mathfrak{F}_i$, there exists a function $L_{d,\kappa} :
	\mathbb{N} \to \mathbb{N} \in \mathfrak{F}_{i + d - 1}$ such that
	$L_{d,\kappa}(t)$ is an upper bound for the length of non-increasing
	sequences from $\left(\mathbb{N}^d\right)^*$ that are $t$-controlled by
	$\kappa$.
\end{lemma}

To conclude, we show how to bound the length of any sequence $s \in S$. By
construction of $H$, $s$ is a function-action non-increasing sequence $f_0 \dots
f_n$ of non-negative functions---except for the last function, which might be
negative. The vector sequence $\vec{f_0} \dots \vec{f}_{n-1}$ is therefore
non-increasing and, for all $0 \le i < n$, the vector $\vec{f_i}$ has dimension
$|Q|+2$ and contains only non-negative numbers.  It follows from
Lemmas~\ref{lem:control} and~\ref{lem:length-function} that the length of the
vector sequence is less than $h(c_0 + |Q| + w_{\max})$, where $h = L_{|Q| + 2,k}
\in \mathfrak{F}_{|Q|+3}$. Hence, the length of $s$ is bounded by $h(c_0 + |Q| +
w_{\max}) + 1$.  Let us write $|G| = |Q| + c_0 + w_{\max} + |\Delta| + |\Obs|$.
We thus have that $h(|G|) + 1$ bounds the length of all $s \in S$. Clearly then,
the size of $S$ is at most $|\Delta|^{h(|G|) + 1}$. More coarsely, we have that
$2^{(h(|G|)+1)^2}$ bounds the size of $S$. It follows from
Lemmas~\ref{lem:props-fast-growing-functions}--\ref{lem:length-function} that
the latter bound is primitive recursive for all fixed $G$. Since $\Ack$
dominates all primitive-recursive functions, we conclude $|S|$ is
$\mathcal{O}(\Ack(|G|))$. As safety games are known to be solvable in linear
time with respect to the size of the game graph (see, e.g.,~\cite{ag11}), the
desired result then follows from Lemma~\ref{lem:safety-energy}.

\begin{theorem}\label{thm:upper}
	The fixed initial credit problem is decidable in Ackermannian time.
\end{theorem}

\section{Lower bound}
In the sequel we will establish Ackermannian hardness of the fixed initial
credit problem, thus giving a negative answer to the question of whether the
problem has a primitive-recursive algorithm. This question is of particular
interest in light of recent work by Jurdzi\'{n}ski et al.~\cite{jls15} in which
it is shown the same problem is \TWOEXP-complete for multi-dimensional games
with full observation.

To begin, we will formally define the \ACK\ complexity class---using the
hierarchies of functions introduced in the previous section. We will then
adapt the translation from Minsky machines to partial-observation energy games
presented in~\cite{ddgrt10} (to argue the unknown initial credit problem is
undecidable) and reduce the existence of a halting run with bounded
counter values in the original machine to the fixed initial credit problem in
the constructed game. Finally, we will describe how to make sure the bound on
the counters is Ackermannian (without explicitly computing the Ackermann
function during the reduction).

\paragraph*{The complexity class}
We adopt here the definition proposed by Schmitz~\cite{schmitz16} for the class
of Ackermannian decision problems:
\[ \textstyle
	\ACK \defeq \bigcup_{g \in \mathfrak{F}_{< \omega}}
	\DTIME\left(\Ack\left(g(n)\right)\right)
\]
where $\mathfrak{F}_{< \omega} \defeq \bigcup_{i \in \mathbb{N}}
\mathfrak{F}_i$.  Note that we allow ourselves any kind of primitive-recursive
reduction. It was shown in~\cite{schmitz16} that: for any two functions $g,f \in
\mathfrak{F}_{<\omega}$, there exists $p$ in $\mathfrak{F}_{<\omega}$ such that
$f \circ \Ack \circ g$ is dominated by $\Ack \circ p$. It follows the distinction
between time-bounded and space-bounded computations is actually irrelevant here
since $F_2$ is already of exponential growth.

\subsection{Minsky machine simulation}\label{sec:2cm-simulation}
A \emph{$2$-counter Minsky machine} ($2$CM) consists of a finite
set of control states $Q$, initial and final states $q_I,q_F \in Q$, a set of
two counters $C$, and a finite set of instructions which act on the counters.
Namely, $inc_k$ increases the value of counter $k$ by $1$, $dec_k$ decreases the
same value by $1$. Additionally, $0?_k$ serves as a \emph{zero-check} on counter
$k$ which blocks if the value of counter $k$ is not equal to $0$. More formally,
the transition relation $\delta$ contains
tuples $(q,\iota,k,q')$ where $q,q' \in Q$ are source and
target states respectively, $\iota$ is an instruction from $\{inc,dec,0?\}$
which is applied to counter $k \in C$. We focus here on deterministic $2$CMs,
i.e. for every state $q \in Q$ either
\begin{itemize}[nolistsep]
	\item $\delta$ has exactly one outgoing transition, which is 
		an increase instruction, i.e.
		$(q,\iota,\cdot,\cdot)$ with $\iota = inc$; or
	\item $\delta$ has two transitions: a decrease
		$(q,dec,k,\cdot)$ and a zero-check $(q,0?,k,\cdot)$ instruction.
\end{itemize}
We write
$|M|$ instead of $|Q|$ to denote the \emph{size of $M$}.  A \emph{configuration
of $M$} is a pair $(q,v)$ of a state $q \in Q$ and a
valuation $v : C \to \mathbb{N}$.\footnote{Note that we consider the variant of
Minsky machines which guards all decreases with zero-checks. Hence, all
counters will have only non-negative values at all times.} A 
\emph{run} of $M$ is a finite sequence $\rho = (q_0,v_0)
\delta_0 \dots \delta_{n-1} (q_n,v_n)$ such that $q_0 = q_I$, $v_0(k) = 0$
for all $k \in C$, and $v_{i+1}$ is the correct valuation
of the counters after applying $\delta_i$ to $v_i$ for all $1 \le i \le
n$. A run $(q_0,v_0) \delta_0 \dots \delta_{n-1} (q_n,v_n)$ is \emph{halting} if
$q_n = q_F$ and it is \emph{$m$-bounded} if $v_i(k) \le
m$ for all $0 \le i \le n$ and all $k \in C$.

\begin{problem}[$f$-Bounded halting problem]
	Given a $2$CM $M$, decide whether $M$ has an
	$f(|M|)$-bounded halting run.
\end{problem}
This bounded version of the halting problem is decidable. However, if $f =
\Ack$, then it is \ACK-complete~\cite{ss12,schmitz16}.

\begin{lemma}\label{lem:short-witness}
	A $2$CM $M$ has an
	$f(|M|)$-bounded halting run if and only if it has an $f(|M|)$-bounded
	halting run of length at most $|M|(f(|M|))^2$.
\end{lemma}
\begin{proof}
	We focus on the only non-trivial direction. If no $f(|M|)$-bounded run
	of the machine does reach $q_f$ in at most $|M|(f(|M|))^2$ steps, then
	we have two possibilities. It could be the case that the machine has a
	counter whose value goes above $f(|M|)$ before reaching $q_F$.
	Clearly, $M$ has no $f(|M|)$-bounded halting run in this case.
	Otherwise, the machine must have repeated at least one configuration.
	Since the machine is deterministic, this means it will never
	halt (i.e. it will cycle without reaching $q_F$).
\end{proof}

For the rest of this section, let us 
consider a fixed $2$CM $M = (Q,q_I,q_F,C,\delta)$.

\paragraph*{Energy game for the $\mathsf{id}$-bounded halting problem}
We will now construct a blind energy game $G_M$ such that $M$ has an
$|M|$-bounded halting run if and only if there is a winning observation-based
strategy for \eve in $G_M$ given initial credit $c_0 = 0$, i.e. we reduce from
the $\mathsf{id}$-bounded halting problem where $\mathsf{id}(x) = x$ is the identity
function. Her winning observation-based strategy will be to
perpetually simulate the halting run of the machine. We set the alphabet
$\Sigma$ of $G_M$ to be the set of transitions of $M$ plus a fresh symbol $\#$,
that is $\Sigma = \delta \cup \{\#\}$.  The game $G_M$ starts with a
non-deterministic transition into one of several gadgets we will describe now.
The weight of the transition into each gadget is shown (labelling the initial
arrow) in the Figures for the gadgets. Each gadget will check the sequence of
letters played by \eve has some specific property, lest some play will get a
negative energy level. Since the game is blind, \eve will not know which gadget
has been chosen and will therefore have to make sure her strategy (infinite
word) has all the properties checked by each gadget.

The gadget depicted in Figure~\ref{fig:gadget1} makes sure that \eve plays $\#$
as her first letter. Indeed, if she plays a strategy
which does not comply then there is a play which will end up in the $-1$ loop
and thus gets a negative energy level. If she plays $\#$, all
plays in this gadget go to the $0$ loop and will never have a negative energy
level. To make sure that after $\#$ \eve plays the first transition from $M$,
and then the second, \dots, we have $|\delta| + 1$
instances of the gadget from Figure~\ref{fig:gadget2}. (Additionally, after
having played a transition leading to $q_F$ in $M$---i.e. of the form
$(\cdot,\cdot,\cdot,q_F)$---she must play $\#$ once more.) The intuition behind
the gadget is simple: if she violates the order of the sequence of transitions
then there will be a play consistent with her strategy which, in the
corresponding gadget, reaches the $-1$ loop. If she plays the letters in the
correct sequence then all plays in all gadgets can only stay in the initial
state or go to the upper $0$ loop.

\begin{figure}
\begin{center}
\begin{tikzpicture}[initial text={$0$},el]
		\node[state,initial above](A){};
		\node[state,left=of A](B){ };
		\node[state,right=of A](C){ };
		
		\path
		(A) edge node[el,swap] {$\#,0$} (B)
		(A) edge node[el] {$\Sigma \setminus \{\#\},0$} (C)
		(B) edge[loop above] node[el] {$\Sigma,0$} (B)
		(C) edge[loop above] node[el] {$\Sigma,-1$} (C)
		;
\end{tikzpicture}
\end{center}
\caption{Gadget which ensures the first letter played by \eve is $\#$.}
\label{fig:gadget1}
\end{figure}
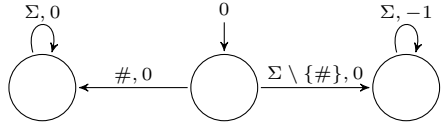

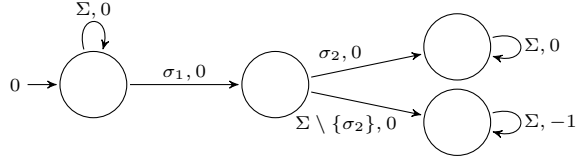
\begin{figure}
\begin{center}
\begin{tikzpicture}[initial text={$0$},el]
		\node[state,initial](A){};
		\node[state,right=of A](B){ };
		\node[state,right=of B,yshift=0.5cm](C){ };
		\node[state,right=of B,yshift=-0.5cm](D){ };

		\path
		(A) edge[loop above] node[el] {$\Sigma,0$} (A)
		(A) edge node[el] {$\sigma_1,0$} (B)
		(B) edge node[el] {$\sigma_2,0$} (C)
		(B) edge node[el,pos=0.8,swap] {$\Sigma \setminus \{\sigma_2\}, 0$}
			(D)
		(C) edge[loop right] node[el] {$\Sigma,0$} (C)
		(D) edge[loop right] node[el] {$\Sigma,-1$} (D)
		;
\end{tikzpicture}
\end{center}
\caption{Gadget which ensures $\sigma_1$ is followed by $\sigma_2$.}
\label{fig:gadget2}
\end{figure}

To verify that \eve plays the letter $\#$ infinitely often, symbolizing the
start of a new simulation of $M$ every time, the game $G_M$ includes the gadget
shown in Figure~\ref{fig:gadget3}.  If \eve plays the
letter $\#$ infinitely often and within $|M|^3$ rounds of each other, all the
plays in the gadget will take the only negatively-weighted transition in the
gadget at most $|M|^3$ times. Hence, all plays in this gadget will never have a
negative energy level. If, however, \eve plays in any other way (eventually
stopping with the letter $\#$ or taking too long to produce it) then there will
be a play which reaches the middle state of the gadget and take the negative
transition enough times to get a negative energy level. In
Figure~\ref{fig:gadget4} we can see a very simple gadget which is instantiated
for each $k \in C$. Its task is to ensure \eve does not play a sequence of
$inc_k$ and $dec_k$ which results in the value of $k$ being larger than $|M|$.

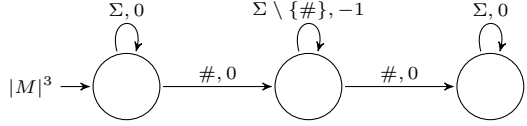
\begin{figure}
\begin{center}
\begin{tikzpicture}[initial text={$|M|^3$},el]
		\node[state,initial](A){ };
		\node[state,right=of A](B){ };
		\node[state,right=of B](C){ };

		\path
		(A) edge[loop above] node[el] {$\Sigma,0$} (A)
		(A) edge node[el] {$\#,0$} (B)
		(B) edge[loop above] node[el] {$\Sigma \setminus
			\{\#\},-1$} (B)
		(B) edge node[el] {$\#,0$} (C)
		(C) edge[loop above] node[el] {$\Sigma,0$} (C)
		;
\end{tikzpicture}
\end{center}
\caption{Gadget which ensures that \eve restarts her simulation of $M$
	infinitely often and that each simulation is of length at most $|M|^3$.}
\label{fig:gadget3}
\end{figure}

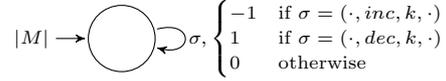
\begin{figure}
\begin{center}
\begin{tikzpicture}[initial text={$|M|$},el]
		\node[state,initial](A){ };

		\path
		(A) edge[loop right] node[el] {$\sigma, \begin{cases}
			-1 & \text{if } \sigma = (\cdot,inc,k,\cdot)\\
			1 & \text{if } \sigma = (\cdot,dec,k,\cdot)\\
			0 & \text{otherwise}
			\end{cases}$} (A)

		;
\end{tikzpicture}
\end{center}
\caption{Gadget which ensures that \eve respects the bound on the counters.}
\label{fig:gadget4}
\end{figure}

Finally, we make sure that \eve does not play a transition which executes a
zero-check incorrectly. To do so, for each $k \in C$ we add to $G_M$ an instance
of the gadget from Figure~\ref{fig:gadget5}.  The intuition of how the gadget
works is as follows. If \eve executes a $0?_k$ instruction when the counter
value is positive (a \emph{zero-cheat}), then there will be a play which goes to
the top-right corner of the gadget and simulates the inverse of the intended
operations on $k$ and moves back to the initial state on the zero-cheat. This
play thus far has an energy level of at most $|M|^3 -1$. If \eve executes a $dec_k$
instruction when the counter value is $0$ (a \emph{positive cheat}), then there
will be a play which goes to the bottom-left corner of the gadget and simulates
the operations on $k$ and moves back to the initial state on the positive cheat.
The latter play also has an energy level of at most $|M|^3 -1$. It follows
that if \eve cheats more than $|M|^3$ times, there will be a play in the gadget
with negative energy level.  If, however, she correctly simulates the zero
checks, then a play can forever stay in the initial state---in which case it
will never have a negative energy level---or it can move to the top-right or
bottom-left corner at some point. There, if \eve is simulating a finite halting
run $\rho$, she will play $\#$ again in at most $|\rho| + 1$ steps. If the play
is still in one those corners at that moment, then we know it moves to the
bottom-right state and that it must have seen a weight of $-1$ at most $|\rho|$
times. Clearly, once there, the play can no longer have a negative energy level.
If the play returned to the initial state before that then, since \eve is not
cheating, the energy level of the play must have been at least $c_0 + |\calM|^3$
and it must have seen a weight of $-1$ at most $|\rho|$ times.

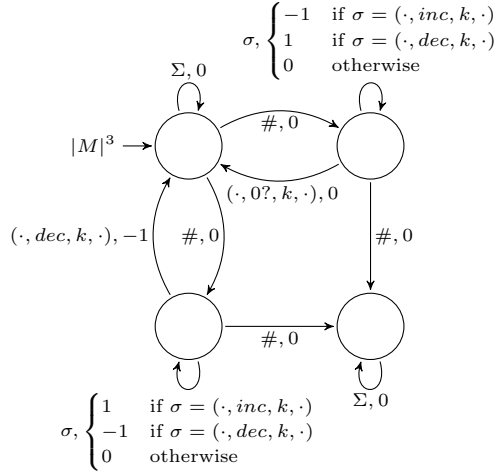
\begin{figure}
\begin{center}
\begin{tikzpicture}[initial text={$|M|^3$},el]
		\node[state,initial](A){ };
		\node[state,right=of A](B){ };
		\node[state,below=of A](C){ };
		\node[state,below=of B](D){ };

		\path
		(A) edge[loop above] node[el] {$\Sigma,0$} (A)
		(A) edge[bend left] node[swap,el] {$\#,0$} (B)
		(B) edge[loop above] node[el] {$\sigma,\begin{cases}
			-1 & \text{if } \sigma = (\cdot,inc,k,\cdot)\\
			1 & \text{if } \sigma = (\cdot,dec,k,\cdot)\\
			0 & \text{otherwise}
			\end{cases}$} (B)
		(B) edge[bend left] node[el,pos=0.48] {$(\cdot,0?,k,\cdot), 0$}
			(A)
		(B) edge node[el] {$\#,0$} (D)
		(D) edge[loop below] node[el] {$\Sigma,0$} (D)
		(C) edge node[el,swap] {$\#,0$} (D)
		(A) edge[bend left] node[el,swap] {$\#,0$} (C)
		(C) edge[bend left] node[el]
			{$(\cdot,dec,k,\cdot),-1$} (A)
		(C) edge[loop below] node[el] {$\sigma, \begin{cases}
			1 & \text{if } \sigma = (\cdot,inc,k,\cdot)\\
			-1 & \text{if } \sigma = (\cdot,dec,k,\cdot)\\
			0 & \text{otherwise}
			\end{cases}$} (C)
		;
\end{tikzpicture}
\end{center}
\caption{Gadget which ensures \eve correctly resolves the guarded decreases:
	executing a $dec_k$ instruction only when $k > 0$ and executing the $0?_k$
	instruction otherwise.}
\label{fig:gadget5}
\end{figure}

\begin{proposition}\label{pro:cm-2-eg}
	$M$ has an $|M|$-bounded halting run if and only if there is a winning
	observation-based strategy for \eve in $G_M$ given initial credit $c_0 =
	0$.
\end{proposition}
\begin{proof}
	If $M$ does have an $|M|$-bounded halting run $\rho$, then we can assume
	that $\rho$ has length at most $|M|^3$ (see
	Lemma~\ref{lem:short-witness}). \eve can then play the blind strategy
	which corresponds to the infinite word $(\# \rho)^\omega$. Since this
	word satisfies all the constraints ensured by the gadgets in $G_M$, no
	play will ever have negative energy level. 
	
	Suppose $M$ has no $|M|$-bounded halting run. \eve cannot play a
	strategy which does not correspond to a valid run of $M$ or there will
	be a play with negative energy level in gadgets $1$, $2$, or $5$. Thus,
	let us assume she does simulate $|M|$ faithfully. We now consider two
	cases depending on whether $|M|$ has a halting run (which is,
	necessarily, not $|M|$-bounded). If $|M|$ has a halting
	run which is not $|M|$-bounded, then a play with negative energy level
	can be constructed in gadget $4$. Similarly, if her simulation of $|M|$
	stays $|M|$-bounded but takes longer than $|M|^3$ steps (because it is
	not halting), the strategy will not be winning because of gadget $3$.
	Hence, she has no observation-based winning strategy.
\end{proof}

We will now generalize the reduction we just presented and use it to prove the
announced \ACK-hardness result. In short, we need to make sure the energy level
of plays entering gadget $4$ becomes $\Ack(|M|)$ and the energy level of plays
entering gadgets $3$ and $5$ becomes greater than $|M|(\Ack(|M|))^2$. The way
in which we propose to do so is to add an initial gadget to $G_M$ which allows
\eve to play a specific sequence of letters to get the required energy
level---and no more---and then forces here into the simulation of $M$ which we
have already described. As a first step, we describe a way of computing $\Ack$
using vectors.

\subsection{Vectorial version of the fast-growing functions}
We will now give a vectorial-based definition of $F_k$, for any $k \in
\mathbb{N}$.
Intuitively, we will use $k + 1$ dimensions to keep track of how
many iterations of $F_j$ (for $0 \le j \le k$) still need to be applied on the
current intermediate value. More formally, for a vector $\mathbf{a} = (
a_k,\dots,a_0 ) \in \mathbb{N}^{k+1}$ we set
\[
	\Phi(\mathbf{a};x) = \Phi(a_k,\dots,a_0;x) \defeq
	F_k^{a_k}(\dots F_1^{a_1}(F_0^{a_0}(x))).
\]
It follows that $\Ack(x) = \Phi(1,\overbrace{0,\dots,0}^{x \text{ times}};x)$.
The following is the key property associated with $\Phi$.
\begin{lemma}\label{lem:key-vectorial}
	For vectors $\mathbf{a},\mathbf{b} \in \mathbb{N}^{k+1}$ and $x,y \in
	\mathbb{N}$, if $\mathbf{a} \le \mathbf{b}$ and $x \le y$
	then $\Phi(\mathbf{a};x) \le \Phi(\mathbf{b};y)$.
\end{lemma}
\begin{proof}
	It is a direct consequence of the definition of $\Phi$ and
	Lemma~\ref{lem:props-fast-growing-functions}.
\end{proof}

We consider the following (family of) rewrite rules $N0$, $N1_j$, and $N2$:
\begin{align*}
	\Phi(\mathbf{a};x) &\to_{N0} x\\
	\Phi(\dots,a_j + 1, a_{j-1}, \dots; x) &\to_{N1_j}
	\Phi(\dots,a_j,x + 1, \dots; x)\\
	\Phi(a_k,\dots,a_0 + 1; x) &\to_{N2} \Phi(a_k,\dots,a_0;x + 1)
\end{align*}
For simplicity, denote the set of rewrite rules $\{N1_j \st 0 < j \le k\}$ by
$N1$. Let us write $(\mathbf{b},y) \leadsto_r (\mathbf{a},x)$ and $(\mathbf{b},y)
\leadsto_N (\mathbf{a},x)$ if rule $r$ or, respectively, a sequence of the rules
$N1$ and $N2$, can be applied to $\Phi(\mathbf{b};y)$ to transform it into
$\Phi(\mathbf{a};x)$.\footnote{Since the rule $N0$ yields a single number, it
cannot be the case that $(\mathbf{b},y) \leadsto_{N0} (\mathbf{a},x)$.}
We remark that $(\mathbf{b},y) \leadsto_N (\mathbf{a},x)$ implies $\mathbf{a}$
is smaller than $\mathbf{b}$ for the lexicographical order. It follows that the
application of the rewrite rules always terminates.
\begin{lemma}\label{lem:termination}
	For any vector $\mathbf{b} \in \mathbb{N}^{k+1}$ and $x \in \mathbb{N}$,
	the set $\{\mathbf{a} \in \mathbb{N}^{k+1} \st \exists y \in \mathbb{N}
	: (\mathbf{b},y) \leadsto_N (\mathbf{a},x)\}$ is finite.
\end{lemma}

Remark that rule $N1_j$ can be applied to $\Phi(\mathbf{a};x)$ for any $0 < j
\le k$ as long as $a_{j} > 0$. We would like to argue that the ``best'' way to
use $N1_j$, in order to obtain the highest possible final value, is to do so
only if all dimensions $0 \le \ell < j$ have value $0$. Formally, let us write
$(\mathbf{b},y) \to_r (\mathbf{a},x)$ if $(\mathbf{b},y) \leadsto_r
(\mathbf{a},x)$ and, additionally, if $r = N1_j$ then it holds that $a_\ell = 0$
for all $0 \le \ell < j$. We then say the application of the rewrite rule $r$
was \emph{proper}. Similarly, we write $(\mathbf{b},y) \to_N (\mathbf{a},x)$ if
$\Phi(\mathbf{a};x)$ can be obtained by proper application of rules $N1$ and
$N2$ to $\Phi(\mathbf{b};y)$.
\begin{lemma}\label{lem:proper-ack}
	For all vectors $\mathbf{a},\mathbf{b} \in \mathbb{N}^{k+1}$ and
	$x,y \in \mathbb{N}$, if $(\mathbf{b},y) \to_N (\mathbf{a},x)$ then
	$\Phi(\mathbf{b};y) = \Phi(\mathbf{a};x)$.
\end{lemma}
\begin{proof}
	Follows directly from the definitions of $\Phi$ and the fast-growing
	functions, and proper application of rules $N1$ and $N2$.
\end{proof}
The above result means that proper application of the rewrite rules to
$\Phi(a_k,\dots,a_0;x)$ give us a correct computation of
$F^{a_k}_k(\dots(F^{a_0}_0(x)))$. We will now argue that improper application of
the rules will result in a smaller value.
\begin{lemma}\label{lem:improper-less-ack}
	For all vectors $\mathbf{a},\mathbf{b} \in \mathbb{N}^{k+1}$ and
	$x,y \in \mathbb{N}$, if $|\mathbf{b}|_{\infty} \le x$ and
	$(\mathbf{b},y) \leadsto_N (\mathbf{a},x)$ then
	$\Phi(\mathbf{b};y) \le \Phi(\mathbf{a};x)$.
\end{lemma}
\begin{proof}
	If $(\mathbf{b},y) \to_N (\mathbf{a},x)$ then the result follows by
	applying Lemma~\ref{lem:proper-ack}. If this is not the case, the
	sequence of rules applied to $\Phi(\mathbf{b};y)$ to obtain
	$\Phi(\mathbf{a};x)$ includes at least one improperly applied rule.
	Since $N2$ cannot be applied improperly, we focus on $N1$.  We will show
	that applying an $N1$ rule improperly cannot increase the value. The
	desired result will then follow by induction.
	
	We will argue that for any $z \in \mathbb{N}$ and any vector $\mathbf{c}
	\in \mathbb{N}^{k+1}$ such that $|\mathbf{c}|_\infty \le z$, it holds
	that for all $2 \le i \le k + 1$, applying any rule from $\{N1_j \st 0 <
	j < i \}$ improperly to $\Phi(\mathbf{c};z)$, yields a value smaller
	than $\Phi(\mathbf{c};z)$.
	For the base case we consider $i = 2$. We need to show the property
	holds for $N1_1$. Assume that $(\mathbf{c},z) \leadsto_{N1_1}
	(c_k,\dots,c_1 - 1,z+1,z)$ and that $c_0 > 0$. By applying $c_0$ times
	the rule $N2$ to $\Phi(\mathbf{c};z)$ we obtain $\Phi(c_k,\dots,c_1,0;z
	+ c_0)$. We can now properly apply $N1_1$ to the latter and obtain
	$\Phi(c_k,\dots,c_1-1,z + c_0;z + c_0)$. It follows from
	Lemma~\ref{lem:key-vectorial} that the claim holds for $i = 2$.  To
	conclude, we show that if it holds for $i$ then it must hold for $i + 1$.
	We only need to show the property holds for $N1_{i}$. Assume that
	$(\mathbf{c},z) \leadsto_{N1_i} (c_k,\dots,c_i-1,z+1,c_{i-2},\dots,z)$
	and that $c_\ell > 0$ for some $0 \le \ell < i$. By applying the
	sequence of rules $N1_\ell N1_{\ell - 1} \dots N1_1 N2 N1_i N1_{i-1}
	\dots N1_1$ to $\Phi(\mathbf{c};z)$ we obtain
	$\Phi(c_k,\dots,c_i-1,z+1,z+1,\dots,z+2;z+1)$. It follows from induction
	hypothesis and the fact that $|\mathbf{b}|_\infty \le z$ that
	$\Phi(\mathbf{c};z)$ is larger than the latter, which is, in turn, larger
	than $\Phi(c_k,\dots,c_i-1,z+1,c_{i-2},\dots;z)$ according to
	Lemma~\ref{lem:key-vectorial}. Thus, the claim holds. 
\end{proof}

In the next section we will detail a new gadget which can be used to pump an energy
level of $m$ up to $F_m(m) = \Ack(m)$. The gadget simulates the rewrite rules to
compute $\Ack$ vectorially.

\subsection{$\Ack$-pumping gadget}
For convenience, we will focus on the blind gadget as a blind energy game
itself. We will later comment on how it fits together with the Minsky machine
game constructed in Section~\ref{sec:2cm-simulation}.

Let us consider a fixed $m \in \mathbb{N}$. The blind energy game $I_m$ we build
has exactly $m + 5$ states, namely: an initial state $q_0$ and states
$\{\top,f,\chi\} \cup \{\alpha_i \st 0 \le i \le m\}$. The game starts with a
non-deterministic choice of state from the set $\{\chi\} \cup \{\alpha_i \st 0
\le i \le m\}$. The weight of the transition going to state $\alpha_m$ is $1$;
to state $\chi$, $m$; to all other states, $0$. The alphabet consists of as many
rewrite rules (defined in the previous section) as required to compute $F_m$.
More formally, the alphabet is $\Sigma = \{N0,N2\} \cup \{N1_j \st 1 \le j \le
m\}$. For clarity, the game has been split into
Figures~\ref{fig:pumping-out}--\ref{fig:pumping-one}, each Figure showing
transitions for different letters. Intuitively, the game $I_m$ allows \eve to
simulate the vectorial computation of $\Ack(m)$. Once she plays the letter $N_0$
the play moves to either $\top$ or $f$, and---as we will argue later---if she
has correctly simulated the computation of the function and the play has reached
$f$, its energy level will be $\Ack(m)$.

\begin{figure}
\begin{center}
\begin{tikzpicture}[node distance=0.8cm]
	\node[state,initial above,
		initial text={\scriptsize{$1$}}](qm){$\alpha_{m}$};
	\node[right=0.3cm of qm](dots1){$\dots$};
	\node[state,initial above,
		initial text={\scriptsize{$0$}},right=0.3cm of
		dots1](qj){$\alpha_j$};
	\node[state,initial above,
		initial text={\scriptsize{$0$}},right=of qj](qj1){$\alpha_{j-1}$};
	\node[state,below=0.8cm of qj1](sink){$\top$};
	\node[right=0.3cm of qj1](dots2){$\dots$};
	\node[state,initial above,
		initial text={\scriptsize{$0$}},
		right=0.3cm of dots2](q0){$\alpha_0$};
	\node[state,initial above,
		initial text={\scriptsize{$m$}},right=of q0](chi){$\chi$};
	\node[state,below=0.8cm of chi](out){$f$};

	\path
	(sink) edge[loop below] node[el]{$\Sigma$,$0$} (sink)
	(out) edge[loop below] node[el]{$\Sigma$,$0$} (out)
	;

	\path
	(qm) edge[bend right] node[el,swap] {$N0$,$0$} (sink)
	(qj) edge node[el,swap] {$N0$,$0$} (sink)
	(chi) edge node[el] {$N0$,$0$} (out)
	(qj1) edge node[el]{$N0$,$0$} (sink)
	(q0) edge[bend left] node[el]{$N0$,$0$} (sink)
	;
\end{tikzpicture}
\end{center}
\caption{Pumping gadget with only transitions for $\sigma = N0$ shown for states
	$\{\chi\}\cup\{\alpha_i \st 0 \le i \le m\}$.}
\label{fig:pumping-out}
\end{figure}
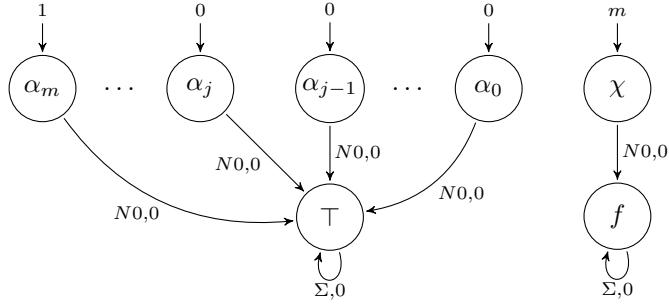

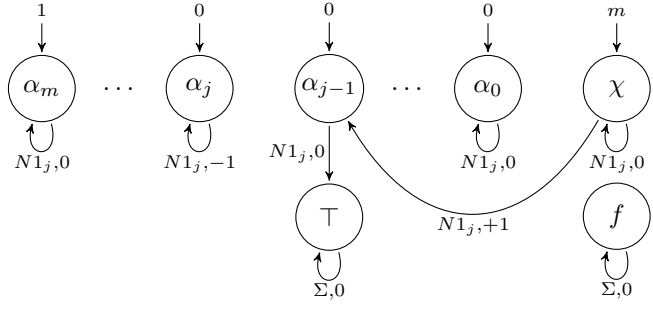
\begin{figure}
\begin{center}
\begin{tikzpicture}[node distance=0.8cm]
	\node[state,initial above,
		initial text={\scriptsize{$1$}}](qm){$\alpha_{m}$};
	\node[right=0.3cm of qm](dots1){$\dots$};
	\node[state,initial above,
		initial text={\scriptsize{$0$}},
		right=0.3cm of dots1](qj){$\alpha_j$};
	\node[state,initial above,
		initial text={\scriptsize{$0$}},right=of qj](qj1){$\alpha_{j-1}$};
	\node[state,below=0.8cm of qj1](sink){$\top$};
	\node[right=0.3cm of qj1](dots2){$\dots$};
	\node[state,initial above,
		initial text={\scriptsize{$0$}},
		right=0.3cm of dots2](q0){$\alpha_0$};
	\node[state,initial above,
		initial text={\scriptsize{$m$}},right=of q0](chi){$\chi$};
	\node[state,below=0.8cm of chi](out){$f$};

	\path
	(sink) edge[loop below] node[el]{$\Sigma$,$0$} (sink)
	(out) edge[loop below] node[el]{$\Sigma$,$0$} (out)
	;
	
	\path
	(qj) edge[loop below] node[el]{$N1_j$,$-1$} (qj)
	(qj1) edge node[el,swap]{$N1_j$,$0$} (sink)
	(chi) edge[out=-120,in=-60,looseness=1.5] node[el]{$N1_j$,$+1$} (qj1)
	(qm) edge[loop below] node[el]{$N1_j$,$0$} (qm)
	(q0) edge[loop below] node[el]{$N1_j$,$0$} (q0)
	(chi) edge[loop below] node[el]{$N1_j$,$0$} (chi)
	;
\end{tikzpicture}
\end{center}
\caption{Pumping gadget with only transitions for $\sigma = N1_j$ shown for states
	$\{\chi\}\cup\{\alpha_i \st 0 \le i \le m\}$.}
\label{fig:pumping-up}
\end{figure}

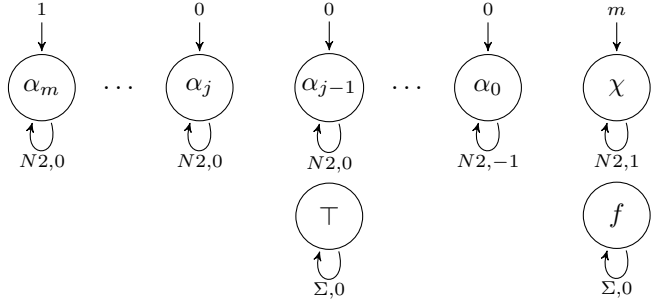
\begin{figure}
\begin{center}
\begin{tikzpicture}[node distance=0.8cm]
	\node[state,initial above,
		initial text={\scriptsize{$1$}}](qm){$\alpha_{m}$};
	\node[right=0.3cm of qm](dots1){$\dots$};
	\node[state,initial above,
		initial text={\scriptsize{$0$}},right=0.3cm of
		dots1](qj){$\alpha_j$};
	\node[state,initial above,
		initial text={\scriptsize{$0$}},right=of qj](qj1){$\alpha_{j-1}$};
	\node[state,below=0.8cm of qj1](sink){$\top$};
	\node[right=0.3cm of qj1](dots2){$\dots$};
	\node[state,initial above,
		initial text={\scriptsize{$0$}},right=0.3cm of
		dots2](q0){$\alpha_0$};
	\node[state,initial above,
		initial text={\scriptsize{$m$}},right=of q0](chi){$\chi$};
	\node[state,below=0.8cm of chi](out){$f$};

	\path
	(sink) edge[loop below] node[el]{$\Sigma$,$0$} (sink)
	(out) edge[loop below] node[el]{$\Sigma$,$0$} (out)
	;
	
	\path
	(qj) edge[loop below] node[el]{$N2$,$0$} (qj)
	(qj1) edge[loop below] node[el]{$N2$,$0$} (qj1)
	(qm) edge[loop below] node[el]{$N2$,$0$} (qm)
	(chi) edge[loop below]
		node[el]{$N2$,$1$} (chi)
	(q0) edge[loop below] node[el]{$N2$,$-1$} (q0)
	;
\end{tikzpicture}
\end{center}
\caption{Pumping gadget with only transitions for $\sigma = N2$ shown for states
	$\{\chi\}\cup\{\alpha_i \st 0 \le i \le m\}$.}
\label{fig:pumping-one}
\end{figure}

Let us write $\Sigma_N$ for the restricted alphabet $ \{N2\} \cup \{N1_j \st 1
\le j \le m\} \subseteq \Sigma$ and $Q_N$ for the set of states $\{\chi\} \cup
\{\alpha_i \st 0 \le i \le m\}$.  We now prove the property this game (or
gadget) enforces. The idea is that \eve playing a sequence of rewrite rules
$\Sigma_N$ has the effect that all plays consistent with her strategy and which
end in $\alpha_i$ have energy level equal to the value of $a_i$ after applying
the rules to $\Phi(\mathbf{a};m)$.

\begin{lemma}\label{lem:correct-plays}
	Consider any play prefix $\pi = q_0 \sigma_0 \dots \sigma_{n-1} q_n$ in
	$I_m$ such that $\sigma_{i} \in \Sigma_N$ for all $0 < i <n$, and
	$\alpha_i \in Q_N$ and $\el(\pi[0..i]) \ge 0$ for all $0 < i \le n$. If
	$q_n = \alpha_j$ then $\el(\pi) = a^{n}_j$, and if $q_n = \chi$ then
	$\el(\pi) = m^n$, where $\mathbf{a^0} = (1,\overbrace{0,\dots,0}^{m
	\text{ times}})$, $m^0 = m$, and $(a^0_m,\dots,a^0_0,m^0)
	\leadsto_{\sigma_0} \dots \leadsto_{\sigma_{n-1}}
	(a^n_m,\dots,a^n_0,m^n)$.
\end{lemma}
\begin{proof}
	We proceed by induction on $n$. Note that if $n = 1$, i.e. the play
	prefix has only two states, the energy level will be equal to $1$ if the
	second state in the play is $\alpha_m$, $m$ if it is $\chi$, and $0$
	otherwise.  Hence the claim holds for prefixes of some length $n$. Let
	us argue that this also holds for prefixes of length $n+1$. Consider
	an arbitrary play prefix $ \pi = q_0 \sigma_0 \dots \sigma_n q_{n+1}$
	for which all assumptions hold. By induction hypothesis we have that
	$\el(\pi[0..n]) = a^n_j$ if $q_n = \alpha_j$ and $m^n$ otherwise. If
	$\sigma_n = N2$, following the transitions shown in
	Figure~\ref{fig:pumping-one} we get that: if $q_n = \alpha_j$, and $0 <
	j \le m$ then the claim holds since $q_{n+1} = \alpha_j$ and $\el(\pi) =
	a^{n+1}_j = a^{n}_j$; also, if $q_n = \alpha_0$ the claim holds since
	$q_{n+1} = \alpha_0$ and $\el(\pi) = a^{n+1}_0 = a^n_0 - 1$ as expected;
	finally, if $q_n = \chi$ then it holds since $q_{n+1} = \chi$ and
	$\el(\pi) = m^{n+1} = m^n + 1$. Otherwise, if $\sigma_n = N1_\ell$,
	following the transitions shown in Figure~\ref{fig:pumping-up} we get
	that: if $q_n = \alpha_j$, and $j \not\in \{\ell,\ell-1\}$ then the claim
	holds since $q_{n+1} = \alpha_j$ and $\el(\pi) = a^{n+1}_j = a^{n}_j$; if
	$q_n = \chi$ then it holds since either $q_{n+1} = \chi$ and $\el(\pi) =
	m^{n+1} = m^n$ or $q_{n+1} = \alpha_{\ell-1}$ and $\el(\pi) =
	a^{n+1}_{\ell-1} = m^n + 1$; it cannot be the case that $q_n =
	\alpha_{\ell-1}$ since otherwise $q_{n+1}$ must be $\bot$ and that would
	violate our assumptions; finally, if $q_n = \alpha_\ell$ then $q_{n+1} =
	\alpha_\ell$ and $\el(\pi) = a^{n+1}_\ell = a^n_\ell - 1$ as required. The
	result thus follows by induction.
\end{proof}

We are finally ready to prove our main result.
\begin{theorem}\label{thm:lower}
	The fixed initial credit problem is \ACK-hard, even for blind games.
\end{theorem}
\begin{proof}
	For any $2$CM $M$ we will consider two instances of the new pumping
	gadget $I_m$, with $m = |M|$, and one instance of the $2$CM-simulating
	game $G_M$ from Section~\ref{sec:2cm-simulation}. (The first copy of
	$I_m$ will be used to compute the Ackermann function while the second
	one will be used to obtain a value greater than $m(\Ack(m))^2$.)
	Additionally, we will add two new states, $s_0$ and $s_1$, which have a
	$0$-weighted self-loop on all letters from the alphabet of $I_m$, except
	for $N0$, and a bad sink state, $\bot$, which has self-loops with weight
	$-1$ on all letters from the alphabets of $I_m$ and $G_M$. The good sink
	$\top$, in copies of $I_m$ now have self-loops with weight $0$ on all
	letters from the alphabet of $G_M$ and $I_m$.
	
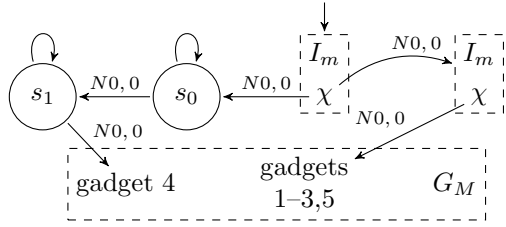
\begin{figure}
\begin{center}
\begin{tikzpicture}
	\node(Im1){$I_m$};
	\node[below=0.2cm of Im1](chi1){$\chi$};
	\node[fit=(Im1) (chi1),initial above,draw,dashed](fit1){ };

	\node[right=1.5cm of Im1](Im2){$I_m$};
	\node[below=0.2cm of Im2](chi2){$\chi$};
	\node[fit=(Im2) (chi2),draw,dashed](fit2){ };

	\node[state,left=1.2cm of chi1](s0){$s_0$};
	\node[state,left=1cm of s0](s1){$s_1$};

	\node[below=0.5cm of s0,xshift=-0.8cm](g4){gadget $4$};
	\node[right=1cm of g4,align=center](gs){gadgets\\$1$--$3$,$5$};
	\node[right=1cm of gs](GM){$G_M$};
	\node[fit=(g4) (gs) (GM),draw,dashed](fit3){ };

	\path
	(chi1) edge[bend left] node[el,pos=0.7]{$N0,0$} (fit2)
	(chi1) edge node[el,swap]{$N0,0$} (s0)
	(s0) edge[loop above] (s0)
	(s0) edge node[el,swap]{$N0,0$} (s1)
	(s1) edge[loop above] (s1)
	(s1) edge node[el]{$N0,0$} (g4)
	(chi2) edge node[el,swap]{$N0,0$} (gs)
	;

\end{tikzpicture}
\end{center}
\caption{Overview of the blind game used to show \ACK-hardness.}
\label{fig:full-construction}
\end{figure}

	We describe how all five components are connected (see
	Figure~\ref{fig:full-construction}). From $\chi$ in the first copy of
	$I_m$ with $N0$ and weight $0$ we non-deterministically go to $s_0$ and
	the initial state of the second $I_m$; from $s_0$ with $N0$ and weight
	$0$ we deterministically go to $s_1$; from $s_1$ with $N0$ and weight
	$0$ we deterministically go to the $G_M$ gadget from
	Figure~\ref{fig:gadget4}; from $\chi$ in the second copy of $I_m$ with
	$N0$ and weight $0$ we non-deterministically go to all copies of the
	gadgets from Figures~\ref{fig:gadget1}, \ref{fig:gadget2},
	\ref{fig:gadget3}, and~\ref{fig:gadget5}. Finally, to make sure the
	transition relation is total, from both copies of $I_m$, $s_0$, and
	$s_1$ we add transitions to $\bot$ on all letters from the alphabet of
	$G_M$. Also, from $G_M$ we add transitions to $\bot$ on all letters from
	the alphabet of $I_m$.

	We will now argue that \eve has an observation-based winning strategy
	for initial credit $c_0 = 0$ if and only if $M$ has an
	$\Ack(|M|)$-bounded halting run.

	\item \paragraph*{If it halts, she wins}
	If $M$ has an $\Ack(|M|)$-bounded halting run, then \eve should play the
	sequence of proper rewrite rules required to compute $m' = \Ack(|M|)$
	vectorially from $\Phi(1,0,\dots,0;m)$.
	She will then play $N0$ and choose letters to compute
	\(
		F^{m'+1}_{m'}(F^{m'}_{m'-1}(\dots F^{m'}_0(m'+m)))
	\)
	which we denote by $m''$.
	Note that $m'' \ge |M|(\Ack(|M|))^2$. Finally, she will simulate the
	halting run of $M$. From Lemma~\ref{lem:correct-plays} we have that:
	after the first time she plays $N0$ the play is either in the first
	copy of $\top$ (and will never have negative energy level) or
	it has energy level $\Ack(|M|)$ and is now in the initial state of the
	second copy of $I_m$ or $s_0$. By playing $N0$ the play now moves to
	$s_1$ or the simulation of the computation of $m''$. After playing a
	third $N0$, the play moves from the second copy of $\chi$---with energy
	level $m''$---and from $s_1$---with $m'$--- to the corresponding
	gadgets in $G_M$ or it moves to $\top$, where \eve cannot lose.
	Hence, any play consistent with this strategy of \eve, and which does not
	enter $G_M$, cannot have negative energy level. If the play has entered
	$G_M$ then the same arguments as presented to prove
	Proposition~\ref{pro:cm-2-eg} should convince the reader that it cannot
	have negative energy level.

	\item \paragraph*{If it does not halt, she does not win} If $M$ has no
	$\Ack(|M|)$-bounded halting run then, from
	Lemmas~\ref{lem:termination} and~\ref{lem:correct-plays} we have
	that \eve eventually play three times $N0$ to exit the copies of $I_m$
	and enter $G_M$---or end in a $\top$ state---lest we can construct a
	play with negative energy level. Also, using
	Lemma~\ref{lem:improper-less-ack}, we conclude that she cannot exit the
	two copies of $I_m$ and enter $G_M$ with energy level greater than
	$\Ack(|M|)$ for the gadget from Figure~\ref{fig:gadget4} or energy level
	greater than $m''$ for the other gadgets, respectively. It thus follows
	from the proof of Proposition~\ref{pro:cm-2-eg} that she has no
	observation-based winning strategy.
\end{proof}

\bibliographystyle{alpha}
\bibliography{mpay-ii}

\end{document}